\documentclass[letterpaper]{article} 
\usepackage{aaai23}  
\usepackage{times}  
\usepackage{helvet}  
\usepackage{courier}  
\usepackage[hyphens]{url}  
\usepackage{graphicx} 
\usepackage{natbib}  
\usepackage{caption} 
\usepackage{algorithm}
\usepackage{algorithmic}

\usepackage{newfloat}
\usepackage{listings}
\DeclareCaptionStyle{ruled}{labelfont=normalfont,labelsep=colon,strut=off} 
\floatstyle{ruled}
\newfloat{listing}{tb}{lst}{}
\floatname{listing}{Listing}
\nocopyright

\usepackage{amsmath,amsthm,amssymb,bm,cleveref,multirow}
\usepackage[usenames,dvipsnames,table]{xcolor}

\newtheorem{theorem}{Theorem}[section]

\newtheorem{lemma}[theorem]{Lemma}

\crefname{assumption}{Assumption}{Assumptions}

\newtheorem{observation}[theorem]{Observation}

\theoremstyle{definition}
\newtheorem{definition}[theorem]{Definition}

\newcommand{\cH}{\mathcal{H}}
\newcommand{\cX}{\mathcal{X}}

\newcommand{\cM}{\mathcal{M}}

\newcommand{\cA}{\mathcal{A}}
\newcommand{\cI}{\mathcal{I}}
\newcommand{\cP}{\mathcal{P}}
\newcommand{\cPc}{\cP^{\mathrm{conn}}}

\newcommand{\eps}{\varepsilon}
\newcommand{\N}{\mathbb{N}}

\newcommand{\R}{\mathbb{R}}

\newcommand{\bu}{{\bm u}}
\newcommand{\tbu}{\tilde{\bu}}
\newcommand{\tu}{\tilde{u}}

\renewcommand{\phi}{\varphi}

\DeclareMathOperator{\Ber}{Ber}

\DeclareMathOperator{\range}{range}

\DeclareMathOperator{\rank}{rank}
\DeclareMathOperator{\scr}{scr}
\DeclareMathOperator{\topk}{top-k}
\DeclareMathOperator{\topv}{top-}

\definecolor{Gred}{RGB}{219, 50, 54}
\definecolor{Ggreen}{RGB}{60, 186, 84}
\definecolor{Gblue}{RGB}{72, 133, 237}
\definecolor{Gyellow}{RGB}{247, 178, 16}
\definecolor{ToCgreen}{RGB}{0, 128, 0}
\definecolor{myGold}{RGB}{231,141,20}
\definecolor{myBlue}{rgb}{0.19,0.41,.65}
\definecolor{myPurple}{RGB}{175,0,124}

\allowdisplaybreaks

\title{Differentially Private Fair Division}
\author{
    Pasin Manurangsi,\textsuperscript{\rm 1}
    Warut Suksompong\textsuperscript{\rm 2}
}
\affiliations{
    \textsuperscript{\rm 1}Google Research\\
    \textsuperscript{\rm 2}National University of Singapore
}

\begin{document}

\maketitle

\begin{abstract}
Fairness and privacy are two important concerns in social decision-making processes such as resource allocation. 
We study privacy in the fair allocation of indivisible resources using the well-established framework of differential privacy.
We present algorithms for approximate envy-freeness and proportionality when two instances are considered to be adjacent if they differ only on the utility of a single agent for a single item.
On the other hand, we provide strong negative results for both fairness criteria when the adjacency notion allows the entire utility function of a single agent to change.
\end{abstract}

\section{Introduction}

Fairness is a principal concern in numerous social decision-making processes, not least when it comes to allocating scarce resources among interested parties.
Whether we divide equipment between healthcare personnel, assign facility time slots to potential users, or distribute office space among working groups in an organization, it is desirable that all parties involved feel fairly treated.
While fair division has been studied in economics for several decades \citep{BramsTa96,Moulin03,Moulin19}, the subject has received substantial interest from computer scientists in recent years, much of which has concentrated on the fair allocation of indivisible resources \citep{BouveretChMa16,Markakis17,Walsh20,Suksompong21,AmanatidisBiFi22,AzizLiMo22}.

The fair division literature typically focuses on satisfying concrete fairness criteria.
Two of the most prominent criteria are \emph{envy-freeness} and \emph{proportionality}.
In an envy-free allocation, no agent prefers to have another agent's bundle instead of her own.
In a proportional allocation, every agent receives value at least $1/n$ of her value for the entire resource, where $n$ denotes the number of agents.
As neither of these criteria is always satisfiable,\footnote{For example, imagine two siblings fighting for one toy.} researchers have proposed the relaxations \emph{envy-freeness up to $c$ items (EF$c$)} and \emph{proportionality up to $c$ items (PROP$c$)}; here, $c$ is a non-negative integer parameter.
Under additive utilities, EF$c$ implies PROP$c$ for every $c$,\footnote{See the proof of Proposition~2.2(b) in our prior work \citep{ManurangsiSu21}.} and an EF1 allocation (which must also be PROP1) is guaranteed to exist \citep{LiptonMaMo04}.

In addition to fairness, another consideration that has become increasingly important nowadays---as large amounts of data are constantly collected, processed, and analyzed---is privacy.
Indeed, an agent participating in a resource allocation procedure 
may not want other participants to know her preferences if she considers these as sensitive information, for example, if these preferences correspond to the times when she is available to use a facility, or if they represent her valuations for potential team members when distributing employees in an organization.
Consequently, a desirable procedure should ensure that an individual participant's preferences cannot be inferred based on the output of the procedure.
Achieving privacy alone is trivial, as the procedure can simply ignore the agents' preferences and always output a fixed allocation that it announces publicly in advance.
However, it is clear that such a procedure can be highly unfair for certain preferences of the agents.
Despite its significance, the issue of privacy has been largely unaddressed in the fair division literature as far as we are aware.\footnote{\citet{SunYa09} studied a setting in which agents have ``reservation values'' over items, and considered a mechanism to be private if it does not require agents to reveal these values.}

In this paper, we investigate the fundamental question of whether fairness and privacy can be attained simultaneously in the allocation of indivisible resources.
We use the well-established framework of \emph{differential privacy (DP)}, which has been widely adopted not only in academia but also in industry~\citep{ErlingssonPiKo14,Shankland14,Greenberg16,Apple17,DingKuYe17} as well as government sectors~\citep{Abowd18}.
Intuitively, the output distribution of a (randomized) DP\footnote{We use the abbreviation DP for both ``differential privacy'' and ``differentially private''.} algorithm should not change by much when a single ``entry'' of the input is modified. 
DP provides a privacy protection for individual entries by ensuring that an adversary with access to the output can only gain limited information about each individual entry. 
At the same time, DP algorithms often still provide useful outputs based on the aggregated information.
We outline the tools and concepts from DP used in this work in \Cref{sec:differential-privacy}.\footnote{For in-depth treatments of the subject, we refer to the surveys by \citet{Dwork08} and \citet{DworkRo14}.}

\renewcommand{\arraystretch}{1.4}
\begin{table*}[!ht]
\centering
\begin{tabular}{|c|c|c|c|}
\cline{3-4}
\multicolumn{2}{c|}{} & \textbf{Agent-Level DP} & \textbf{(Agent $\times$ Item)-Level DP} \\
\hline
\multirow{2}*{EF} & Upper & $O(m/n)$ (Trivial) & $O(n \log m)$ (\Cref{thm:alg-ef-em}) \\
\cline{2-4}
& Lower & $\Omega\left(\sqrt{\frac{m}{n} \log n}\right)$ (\Cref{thm:agent-level-ef-lb}) & $\Omega(\log m)$ (\Cref{thm:lb-agent-item-dp-connected}) \\
\hline
\multirow{2}*{PROP} & Upper & $O(m/n)$ (Trivial) & $O(\log m)$ (\Cref{thm:prop-moving-knife}) \\
\cline{2-4}
& Lower & $\Omega\left(\sqrt{\frac{m}{n}}\right)$ (\Cref{thm:agent-level-prop-lb}) & $\Omega((\log m)/n)$ (\Cref{thm:lb-agent-item-dp-connected-prop}) \\
\hline
\end{tabular}
\caption{Overview of our results. We display the smallest $c$ such that there is an $\eps$-DP algorithm that outputs an EF$c$ or PROP$c$ allocation with probability at least $1 - \beta$. 
For simplicity of the bounds, we assume that $m > n^2$, $\eps$ is a small constant, and $\beta$ is sufficiently small (depending only on $\eps, n$ and not on $m$). 
All upper bounds hold even for connected allocations. Lower bounds for (agent $\times$ item)-level DP hold only for connected allocations, but those for agent-level DP hold for arbitrary allocations.
The bounds $O(m/n)$ follow trivially from outputting a fixed allocation in which each agent receives $O(m/n)$ items.
} 
\label{table:summary}
\end{table*}

As alluded to above, DP is defined with respect to what one considers to be an entry of the input, or equivalently, in terms of adjacent inputs.
We consider two notions of adjacency between fair division instances.
For \emph{agent-level adjacency}, two instances are considered to be adjacent if they differ on the utility function of at most one agent.
For \emph{(agent $\times$ item)-level adjacency}, two instances are adjacent if they differ at most on the utility of a single agent for a single item.
We work with the standard definition of \emph{$\eps$-differential privacy ($\eps$-DP)}: for a parameter $\eps\ge 0$, an algorithm is said to satisfy $\eps$-DP if the probability that it outputs a certain allocation for an input and the corresponding probability for an adjacent input differ by a factor of at most $e^\eps$. 
Note that, for the same $\eps$, \emph{agent-level DP} offers a stronger privacy protection for an entire utility function of an individual agent, whereas \emph{(agent $\times$ item)-level DP} only offers a protection for a utility of a single agent for a single item.
Our goal is to devise $\eps$-DP algorithms that output an EF$c$ or PROP$c$ allocation for a small value of $c$ with sufficiently high probability, or to prove that this task is impossible.
Denote by $n$ and $m$ the number of agents and items, respectively.

We begin in \Cref{sec:agent-level} by considering agent-level DP. 
For this demanding benchmark, we establish strong lower bounds with respect to both approximate envy-freeness and proportionality (\Cref{thm:agent-level-ef-lb,thm:agent-level-prop-lb}). 
In both cases, our lower bounds imply that, for fixed $n$ and $\eps$, no $\eps$-DP algorithm can output an EF$c$ or PROP$c$ allocation for $c = o(\sqrt{m})$ with some large constant probability.
Our results hold even when the agents have binary additive utilities, and indicate that agent-level DP is too stringent to permit algorithms with significant fairness guarantees. 

In \Cref{sec:agent-item-level}, we turn our attention to (agent $\times$ item)-level DP, and deliver encouraging news for this more relaxed notion. 
In contrast to the previous lower bounds, we present $\eps$-DP algorithms for EF$c$ and PROP$c$ where $c$ only grows logarithmically in $m$ for fixed $n$ and $\eps$ (\Cref{thm:alg-ef-em,thm:prop-moving-knife}). 
Our EF$c$ algorithm works for arbitrary monotone utility functions, whereas our PROP$c$ algorithm allows (not necessarily binary) additive utilities.
Moreover, our algorithms always output allocations that are \emph{connected}.\footnote{See \Cref{sec:prelim} for the definition. 
Connectivity can be desirable when there is a spatial or temporal order of the items, for instance, when allocating time slots to facility users or offices along a corridor to research groups \citep{BouveretCeEl17,Suksompong19,BeiIgLu22,BiloCaFl22}.} 
We complement these results by showing a tight lower bound of $\Omega(\log m)$ for connected allocations (\Cref{thm:lb-agent-item-dp-connected,thm:lb-agent-item-dp-connected-prop}), even with binary additive utilities. 

A summary of our results can be found in \Cref{table:summary}.

\section{Preliminaries}
\label{sec:prelim}

\subsection{Fair Division}
\label{sec:fair-division}

In fair division of indivisible items, there is a set $N = [n]$ of agents and a set $M = [m]$ of items, where $[k]$ denotes the set $\{1,2,\dots,k\}$ for each positive integer $k$.
The utility function of agent~$i$ is given by $u_i: 2^M \to \R_{\geq 0}$. 
Throughout this work, we assume that utility functions are \emph{monotone}, that is, $u_i(S) \le u_i(T)$ for any $S\subseteq T\subseteq M$. 
For a single item $j\in M$, we write $u_i(j)$ instead of $u_i(\{j\})$.
We seek to output an \emph{allocation} $A = (A_1, \dots, A_n)$, which is an ordered partition of $M$ into $n$ bundles. 

We consider two important fairness notions. 
Let $c$ be a non-negative integer.
\begin{itemize}
\item An allocation $A$ is said to be \emph{envy-free up to $c$ items (EF$c$)} if, for any $i, i' \in N$, there exists $S\subseteq {A_{i'}}$ with $|S| \le c$ such that $u_i(A_i) \geq u_i(A_{i'} \setminus S)$.
\item An allocation $A$ is said to be \emph{proportional up to $c$ items (PROP$c$)} if, for any $i \in N$, there exists $S\subseteq M\setminus A_i$ with $|S| \le c$ such that $u_i(A_i) \geq u_{i}(M) / n - u_i(S)$.
\end{itemize}

We say that an allocation~$A$ is \emph{connected} if each $A_i$ corresponds to an interval, i.e., $A_i = \{\ell, \ell + 1, \dots, r\}$ for some $\ell, r \in M$; it is possible that $A_i$ is a singleton or empty. 
Let $\cPc(m, n)$ denote the set of all connected allocations.
We will use the following result on the existence of connected EF2 allocations.\footnote{Recently, \citet{Igarashi23} improved this guarantee to EF1.}
\begin{theorem}[\cite{BiloCaFl22}] \label{thm:connected-ef2}
For any $m, n \in \N$ and any monotone utility functions, there exists a connected EF2 allocation.
\end{theorem}

A utility function $u_i$ is \emph{additive} if $u_i(S) = \sum_{j \in S} u_i(j)$ for all $S \subseteq M$.
Furthermore, an additive utility function is said to be \emph{binary} if $u_i(j) \in \{0, 1\}$ for all $j \in M$.

\subsection{Differential Privacy}
\label{sec:differential-privacy}

Let us start by recalling the general definition of DP.
Denote by $\cX$ the set of all possible inputs to the algorithm.

\begin{definition}[Differential Privacy~\cite{DworkMcNi06}] \label{def:dp}
Let $\eps \geq 0$ be a non-negative real number.
A randomized algorithm $\cM$ is said to be \emph{$\eps$-differentially private ($\eps$-DP)} if, for every pair of adjacent inputs $X, X'$, it holds that
\begin{align*}
\Pr[\cM(X) = o] \leq e^\eps \cdot \Pr[\cM(X') = o]
\end{align*}
for all $o \in \range(\cM)$.
\end{definition}

An input in our fair division context consists of the agents' utility functions.
Different notions of adjacency lead to different levels of privacy protection. 
We consider two natural notions: agent-level DP and (agent $\times$ item)-level DP.

\begin{definition}[Agent-Level DP]
Two inputs $(u_i)_{i \in N}$ and $(u'_i)_{i \in N}$ are said to be \emph{agent-level adjacent} if they coincide on all but a single agent, i.e., there exists $i^* \in N$ such that $u_i = u'_i$ for all $i \in N \setminus \{i^*\}$.

An algorithm that is $\eps$-DP against this adjacency notion is said to be \emph{agent-level $\eps$-DP}.
\end{definition}

\begin{definition}[(Agent $\times$ Item)-Level DP]
Two inputs $(u_i)_{i \in N}$ and $(u'_i)_{i \in N}$ are said to be \emph{(agent $\times$ item)-level adjacent} if they coincide on all but the utility of a single agent for a single item, i.e., there exist $i^* \in N, j^* \in M$ such that
\begin{itemize}
\item $u_i = u'_i$ for all $i \in N \setminus \{i^*\}$, and
\item $u_{i^*}(S) = u'_{i^*}(S)$ for all $S \subseteq M \setminus \{j^*\}$.
\end{itemize}

An algorithm that is $\eps$-DP against this adjacency notion is said to be \emph{(agent $\times$ item)-level $\eps$-DP}.
\end{definition}

It is clear that agent-level DP is a more demanding notion than (agent $\times$ item)-level DP.
Specifically, the former provides a stronger privacy protection than the latter, and designing an algorithm for the former is more difficult. 
Indeed, we will prove strong lower bounds for agent-level DP and present algorithms for (agent $\times$ item)-level DP.

Next, we outline several tools from the DP literature that will be useful for our proofs.

\paragraph{Basic Composition.}

The first tool that we will use is the composition of DP: the result of running multiple DP algorithms remains DP, but with a worse privacy parameter. 

\begin{theorem}[Basic Composition of DP, e.g.,~\cite{DworkRo14}] \label{thm:basic-comp}
An algorithm that is a result of running two algorithms (possibly in an adaptive manner) that are $\eps_1$-DP and $\eps_2$-DP, respectively, is $(\eps_1 + \eps_2)$-DP.
\end{theorem}

A special case of basic composition (\Cref{thm:basic-comp}) that is often used is the case $\eps_2 = 0$, which is referred to as \emph{post-processing} of DP.

\begin{observation}[Post-processing of DP] \label{obs:post-processing}
An algorithm that runs an $\eps$-DP subroutine and then returns a function of the output of this subroutine is also $\eps$-DP.
\end{observation}

\paragraph{Parallel Composition.}

While basic composition provides a simple way to account for the privacy budget $\eps$ when we run multiple algorithms, it can be improved in certain cases. One such case is when the algorithms are run on ``disjoint pieces'' of the input. In this case, we do not need to add up the $\eps$'s, a fact known as \emph{parallel composition} of DP~\cite{McSherry10}. Since the statement we use below is slightly different from that in McSherry's work, we provide a full proof for completeness in \Cref{app:proof-parallel-comp}.

\begin{theorem}[\cite{McSherry10}] \label{thm:parallel-comp}
Let $\sim$ be an adjacency notion, and let $\Gamma: \cX \to \cX^k$ be a function such that, if $X \sim X'$, then there exists $i^* \in [k]$ such that $\Gamma(X)_i = \Gamma(X')_i$ for all $i \in [k] \setminus \{i^*\}$, and $\Gamma(X)_{i^*} \sim \Gamma(X')_{i^*}$. 
If $\cM$ is an $\eps$-DP algorithm with respect to the adjacency notion $\sim$, then the algorithm $\cM'$ that outputs $(\cM(\Gamma(X)_1), \dots, \cM(\Gamma(X)_k))$ is also $\eps$-DP with respect to $\sim$.
\end{theorem}

\paragraph{Group Privacy.}
While differential privacy offers protection primarily against the adjacency notion for which it is defined, it also offers protection against more general adjacency notions. 
Below we state one such protection, which is often referred to as \emph{group differential privacy}.

Let $\sim$ be any adjacency relation. 
For $k \in \N$, let us define $\sim_k$ as the adjacency relation where $X \sim_k X'$ if and only if there exists a sequence $X_0, \dots, X_k$ such that $X_0 = X, X_k = X'$, and $X_{i - 1} \sim X_{i}$ for all $i \in [k]$.

\begin{lemma}[Group Differential Privacy, e.g.,~\cite{Vadhan17}] \label{lem:group-dp}
Let $k\in\N$.
If an algorithm is $\eps$-DP with respect to an adjacency notion $\sim$, it is also $(k\eps)$-DP with respect to the adjacency notion $\sim_k$.
\end{lemma}

As an immediate consequence of \Cref{lem:group-dp}, any (agent $\times$ item)-level $\eps$-DP algorithm is also agent-level $(m\eps)$-DP. 
However, the factor $m\varepsilon$ makes the latter guarantee rather weak, especially as $m$ grows.

\paragraph{Sensitivity.} We next define the \emph{sensitivity} of a function, which will be used multiple times in this work. 
Note that the definition depends on the adjacency notion, but we do not explicitly include it in the notation for convenience.

\begin{definition}
The \emph{sensitivity} of a function $f: \cX \to \R$ (with respect to adjacency notion $\sim$) is defined as $\Delta(f) := \max_{X \sim X'} |f(X) - f(X')|$.
\end{definition}

Sensitivity is a key notion in DP.
As shown by~\citet{DworkMcNi06}, the \emph{Laplace mechanism}---which outputs $f(X) + Z$ where $Z$ is drawn from the Laplace distribution\footnote{The Laplace distribution with scale $b$ is the distribution whose probability density function is proportional to $\exp(-|x|/b)$.} with scale $(\Delta(f) / \eps)$---satisfies $\eps$-DP. 
This means that if a function has low sensitivity, then we can estimate it to within a small error (with high probability).

\paragraph{Sparse Vector Technique.}
We will use the so-called \emph{sparse vector technique (SVT)}. 
The setting is that there are low-sensitivity functions $f_1, \dots, f_H: \cX \to \R$. We want to find the first function $f_i$ whose value is above a target threshold. 
A straightforward approach would be to add Laplace noise to each $f_i(X)$ and then select the function accordingly; due to the basic composition theorem, this would require us to add noise of scale $O(H/\eps)$ to each function. 
SVT allows us to reduce the dependency on $H$ to merely $O(\log H)$. 
The technique was first introduced by \citet{DworkNaRe09}, and the convenient version below is due to~\citet[Theorem 3.24]{DworkRo14}.

\begin{theorem} \label{thm:svt}
There exists a constant $\upsilon > 0$ such that the following holds.
Let $f_1, \dots, f_H: \cX \to \R$ be functions with $\Delta(f_1), \dots, \Delta(f_H) \leq 1$. 
For any $\eps > 0$, $\beta\in(0,1)$, and $\tau \in \R$, there exists an $\eps$-DP algorithm such that, if $\max_h f_h(X) \geq \tau$, then, with probability at least $1 - \beta$, the algorithm outputs $h^* \in [H]$ with the following properties:
\begin{itemize}
\item $f_{h^*}(X) \geq \tau - \upsilon \cdot \log(H/\beta) / \eps$;
\item For all $h' < h^*$, $f_{h'}(X) \leq \tau + \upsilon \cdot \log(H/\beta) / \eps$. 
\end{itemize}
\end{theorem}

\paragraph{Exponential Mechanism.}

We will also use the \emph{exponential mechanism (EM)} of \citet{McSherryTa07}.
In its generic form, EM allows us to select a solution from a candidate set $\cH$. Specifically, we may define (low-sensitivity) scoring functions $\scr_h: \cX \to \R$ for each $h \in \cH$. Then, EM outputs a solution that approximately maximizes the score. The precise statement is given below.\footnote{The formulation here can be derived, e.g., by plugging $t = \log(1/\beta)$ into Corollary~3.12 of \citet{DworkRo14}.}

\begin{theorem}[\cite{McSherryTa07}] \label{thm:EM}
For any $\eps > 0$ and $\beta\in(0,1)$, a finite set $\cH$, and a set of scoring functions $\{\scr_h\}_{h \in \cH}$ such that $\Delta(\scr_h) \leq 1$ for each $h \in \cH$, there is an $\eps$-DP algorithm that, on every input~$X$, outputs $h^*$ such that
\begin{align*}
\scr_{h^*}(X) \geq \max_{h \in \cH} \scr_h(X) - \frac{2\log(|\cH| / \beta)}{\eps}
\end{align*}
with probability at least $1 - \beta$.
\end{theorem}

\subsection{Anti-Concentration Inequalities}

Denote by $\Ber(1/2)$ the distribution that is $0$ with probability $1/2$, and $1$ otherwise. The following type of anti-concentration inequalities is well-known; for completeness, we provide a proof of this version in \Cref{app:proof-littlewood-offord}.

\begin{lemma} \label{lem:littlewood-offord-simple}
If $k \geq 100$ and $X_1, \dots, X_k$ are independent random variables drawn from $\Ber(1/2)$, then 
\[
\Pr\left[\sum_{i=1}^k X_i < \frac{k}{2} - 0.1\sqrt{k}\right] \geq \frac{1}{4}.
\]
\end{lemma}

We will also use the following anti-concentration inequality, whose proof can be found in \Cref{app:proof-anti-concen}.

\begin{lemma} \label{lem:anti-concen}
Let $k \geq 100$ and let $X_1, \dots, X_k$ be independent $\Ber(1/2)$ random variables.
Also, let $S := X_1 + \dots + X_k$ and $\gamma \in [2, 2^{k/4}]$. Then,
\begin{align*}
\Pr\left[S > \frac{k}{2} + 0.1\sqrt{k \log \gamma}\right]  \geq \frac{0.1}{\gamma}.
\end{align*}
\end{lemma}

\section{Agent-Level DP}
\label{sec:agent-level}

We begin by considering the demanding notion of agent-level DP, and provide strong negative results for this notion. 

For EF$c$, we show a lower bound that,\footnote{Unless specified otherwise, $\log$ refers to the natural logarithm.} when $m > n\log n$, holds even against $c = \Theta\left(\sqrt{\frac{m}{n} \log n}\right)$. 

\begin{theorem} \label{thm:agent-level-ef-lb}
There exists a constant $\zeta > 0$ such that, for any $\eps > 0$, there is no agent-level $\eps$-DP algorithm that, for any input binary additive utility functions, outputs an EF$c$ allocation with probability higher than $1 - \frac{e^{-\eps}}{200}$, where $c = \left\lfloor\zeta \sqrt{\frac{m}{n} \cdot \min\left\{\log n, \frac{m}{n}\right\}}\right\rfloor$.
\end{theorem}


For proportionality, we prove a slightly weaker bound where $c = \Theta(\sqrt{m/n})$ and the ``failure probability'' required for the lower bound to apply is also smaller at $O_\eps(1/n)$ (compared to $O_\eps(1)$ for envy-freeness).

\begin{theorem} \label{thm:agent-level-prop-lb}
There exists a constant $\zeta > 0$ such that, for any $\eps > 0$, there is no agent-level $\eps$-DP algorithm that, for any input binary additive utility functions, outputs a PROP$c$ allocation with probability higher than $1 - \frac{e^{-\eps}}{8n}$, where $c = \lfloor\zeta \sqrt{m/n}\rfloor$.
\end{theorem}

We first prove \Cref{thm:agent-level-prop-lb}, before proceeding to present the proof of \Cref{thm:agent-level-ef-lb}, which uses similar arguments but requires a more delicate anti-concentration inequality.

\subsection{Proof of \Cref{thm:agent-level-prop-lb}} \label{sec:agent-level-prop-lb}

We let $\zeta = 0.01$.
If $m < 100n$, then $c = 0$ and the theorem holds trivially even without the privacy requirement.
Hence, we may assume that $m \ge 100n$.
Throughout the proof, we consider random utility functions $\bu = (u_i)_{i \in N}$ where each $u_i(j)$ is an independent $\Ber(1/2)$ random variable. 
For brevity, we will not repeatedly state this in the calculations below.

We start by proving the following auxiliary lemma that if $A_i$ is small, then, for a random utility $\bu$ as above, the allocation fails to be PROP$c$ for agent $i$ with a constant probability.

\begin{lemma} \label{lem:agent-level-prop-lb}
For $\zeta = 0.01$, let $c$ be as in \Cref{thm:agent-level-prop-lb} and $A$ be an allocation such that $|A_i| \leq m/n$. Then, we have
\begin{align*}
\Pr_{\bu}[A \text{ is not PROP}c \text{ for agent } i] \geq 1/8.
\end{align*}
\end{lemma}

\begin{proof}
Let $c' = 2c$.
We have
\begin{align}
\Pr_{\bu}&[A \text{ is not PROP}c  \text{ for agent } i] \nonumber \\
&\geq \Pr_{u_i}\left[u_i(A_i) < \frac{u_i(M)}{n} - c\right] \nonumber \\
&= \Pr_{u_i}\left[u_i(A_i) < \frac{u_i(M \setminus A_{i})}{n - 1} - \frac{n}{n-1} \cdot c\right] \nonumber \\
&\geq \Pr_{u_i}\left[u_i(A_i) < \frac{m}{2n} - c' \,\wedge\, u_i(M \setminus A_{i}) \geq \frac{m(n - 1)}{2n}\right] \nonumber \\
&= \Pr_{u_i}\left[u_i(A_i) < \frac{m}{2n} - c'\right] \nonumber \\
&\quad\cdot \Pr_{u_i}\left[u_i(M \setminus A_{i}) \geq \frac{m(n - 1)}{2n}\right] \nonumber \\
&\geq \frac{1}{2} \cdot \Pr_{u_i}\left[u_i(A_i) < \frac{m}{2n} - c'\right], 
\label{eq:prop-c-lb-ind}
\end{align}
where the last inequality follows from the fact that $|M \setminus A_i| \geq m(n - 1)/n$ and symmetry.

Since $|A_i|$ is an integer, $|A_i| \le \lfloor m/n\rfloor$.
Moreover, since $\lfloor m/n\rfloor \ge 100$ and the function $f(k) = k/2 - 0.1\sqrt{k}$ is increasing in $[1,\infty)$, applying \Cref{lem:littlewood-offord-simple} with $k = \lfloor m/n\rfloor$ gives $\Pr_{u_i}\left[u_i(A_i) < \frac{m}{2n} - c'\right] \geq 1/4$. 
Plugging this back into \eqref{eq:prop-c-lb-ind} yields the desired bound. 
\end{proof}

We are now ready to prove \Cref{thm:agent-level-prop-lb}.

\begin{proof}[Proof of \Cref{thm:agent-level-prop-lb}]
Let $\zeta = 0.01$ and let $\cM$ be any agent-level $\eps$-DP algorithm. 
Consider the input utility functions $\bu' = (u'_i)_{i \in N}$ where the utility functions are all-zero, and consider the distribution $\cM(\bu')$. 
For any allocation $A$, we have $\Pr_{i \in N}\left[|A_i| \leq m/n\right] \geq 1/n$ since at least one bundle $A_i$ must have size at most $m/n$.
This implies that $\Pr_{i \in N, A \sim \cM(\bu')}[|A_i| \leq m/n] \geq 1/n$. 
Thus, there exists $i^* \in N$ such that $\Pr_{A \sim \cM(\bu')}[|A_{i^*}| \leq m/n] \geq 1/n$.

Recalling the definition of $\bu$ from earlier and applying \Cref{lem:agent-level-prop-lb}, we have\footnote{Here, PROP$c$ is with respect to $\bu$.}
\begin{align*}
&\Pr_{\bu, A \sim \cM(\bu')}[A \text{ is not PROP}c \text{ for agent } i^*] \\
&\geq \Pr_{A \sim \cM(\bu')}\left[|A_{i^*}| \leq m/n\right] \\
&\quad\cdot\Pr_{\bu, A \sim \cM(\bu')}[A \text{ is not PROP}c \text{ for agent } i^* \mid
|A_{i^*}| \leq m/n] \\
&\geq (1/n) \cdot (1/8) = 1/(8n).
\end{align*}
Hence, there exists $u^*_{i^*}$ such that\footnote{The abbreviation ``w.r.t.'' stands for ``with respect to''.} 
$$
\Pr_{A \sim \cM(\bu')}[A \text{ is not PROP}c \text{ for agent } i^* \text{ w.r.t. } u^*_{i^*}] \geq 1/(8n).
$$ 

Now, let $\bu^*$ be the input utility such that $u^*_i$ is all-zero for each $i \ne i^*$ while $u^*_{i^*}$ is as above. 
Notice that $\bu^*$ is adjacent to $\bu'$ under agent-level adjacency. 
Thus, applying the $\eps$-DP guarantee of $\cM$, we get
\begin{align*}
&\Pr_{A \sim \cM(\bu^*)}[A \text{ is not PROP}c \text{ for agent } i^* \text{ w.r.t. } u^*_{i^*}] \\
&\geq e^{-\eps} \cdot \Pr_{A \sim \cM(\bu')}[A \text{ is not PROP}c \text{ for agent } i^* \text{ w.r.t. } u^*_{i^*}] \\
&\geq (e^{-\eps})/(8n).
\end{align*}
This completes the proof.
\end{proof}

\subsection{Proof of \Cref{thm:agent-level-ef-lb}}

As in the proof of \Cref{thm:agent-level-prop-lb}, we let $\zeta = 0.01$.
If $m < 100n$, then $c = 0$ and the theorem holds trivially even without the privacy requirement.
Hence, we may assume that $m \ge 100n$.
We consider random utility functions $\bu = (u_i)_{i \in N}$ where each $u_i(j)$ is an indepedent $\Ber(1/2)$ random variable. 
For brevity, we will not repeatedly state this in the calculations below.

For an allocation $A$ and an agent $i \in N$, we let $\rank(i; A)$ be the size of the set $\{j \in N \setminus \{i\} \mid |A_j| \geq |A_i|\}$. 

\begin{lemma} \label{lem:agent-level-ef-lb}
For $\zeta = 0.01$, let $c$ be as in \Cref{thm:agent-level-ef-lb} and $A$ be any allocation, and let $\ell\in N$ be such that $\rank(\ell; A) \geq (n - 1)/2$. 
Then, we have
\begin{align*}
\Pr_{\bu}[A \text{ is not EF}c \text{ for agent } \ell] \geq 0.01.
\end{align*}
\end{lemma}

\begin{proof}
Without loss of generality, we may rename the agents so that agent~$\ell$ in the lemma statement becomes agent~$i$ and $|A_1| \geq \cdots \geq |A_i| > |A_{i + 1}| \geq \cdots \geq |A_n|$. 
We now have $\rank(i; A) = i - 1 \geq \lceil (n-1)/2 \rceil$.

We have
\begin{align*}
&\Pr_{\bu}[A \text{ is not EF}c \text{ for agent } i] \\
&\geq \Pr_{u_i}[u_i(A_i) \leq |A_i| / 2] \\
& \qquad\cdot \Pr_{u_i}[A_i \text{ is not EF}c \text{ for agent } i \mid u_i(A_i) \leq |A_i| / 2] \\
&\geq \frac{1}{2} \cdot \Pr_{u_i}[A_i \text{ is not EF}c \text{ for agent } i \mid u_i(A_i) \leq |A_i| / 2] \\
&\geq \frac{1}{2} \cdot \Pr_{u_i}[\exists j < i, u_i(A_j) > |A_i| / 2 + c  \mid u_i(A_i) \leq |A_i| / 2] \\
&= \frac{1}{2} \cdot \left(1 - \prod_{j \in [i - 1]} \Pr_{u_i}[u_i(A_j) \leq |A_i| / 2 + c]\right)
\end{align*}
where the second inequality follows from the fact that $u_i(A_i)$ is a sum of $|A_i|$ independent $\Ber(1/2)$ random variables and the last equality follows from independence.

We now consider two cases.
\begin{itemize}
\item \underline{Case I}: $|A_i| < m / (2n)$. 
Since there are $m$ items and $n$ bundles, we have $|A_1| \geq m/n$.
Our assumption that $m \ge 100n$ and our choice of $c$ imply that $|A_i|/2 + c \leq (|A_1| - 1)/2$. 
Hence, we have $\Pr[u_i(A_1) \leq |A_i|/2 + c] \leq 1/2$. 
It follows that
\begin{align*}
\frac{1}{2} \cdot &\left(1 - \prod_{j \in [i - 1]} \Pr_{\bu}[u_i(A_j) \leq |A_i| / 2 + c]\right) \\
&\geq \frac{1}{2} \cdot \frac{1}{2} = \frac{1}{4}.
\end{align*} 
\item \underline{Case II}: $|A_i| \geq m / (2n)$. 
For any $j \in [i - 1]$, let $S_j \subseteq A_j$ be any subset of $A_j$ of size $|A_i|$. 
Applying \Cref{lem:anti-concen} with $\gamma = \min\{n, 2^{m/(8n)}\}$ and $k = |A_i|$ yields
\begin{align*}
\Pr_{\bu}&[u_i(A_j) > |A_i|/2 + c] \\
&\geq \Pr_{\bu}[u_i(S_j) > |A_i|/2 + c] \\
&\geq \Pr_{\bu}\left[u_i(S_j) > |A_i|/2 + 0.1\sqrt{|A_i| \log\gamma}\right] \\
&\overset{\text{(\Cref{lem:anti-concen})}}{\geq} 0.1/\gamma \geq 0.1/n,
\end{align*}
where the second inequality follows from our choice of parameters $c,\gamma$ and the fact that $|A_i| \ge m/(2n)$.
Thus, we have
\begin{align*}
\Pr_{\bu}&[A_i \text{ is not EF}c \text{ for agent } i] \\
&\geq \frac{1}{2} \cdot \left(1 - \left(1 - \frac{0.1}{n}\right)^{i-1}\right) \\
&\geq \frac{1}{2} \cdot \left(1 - \left(1 - \frac{0.1}{n}\right)^{\left\lceil \frac{n-1}{2} \right\rceil}\right) \\
&\geq \frac{1}{2} \cdot \left(1 - e^{-\frac{0.1}{n}\cdot\left\lceil \frac{n-1}{2} \right\rceil}\right) \\
&\geq \frac{1}{2} \cdot \left(1 - e^{-\frac{0.1}{3}}\right) \geq 0.01.
\end{align*}
\end{itemize}
Hence, in both cases, we have
\begin{align*}
\Pr_{\bu}[A_i \text{ is not EF}c \text{ for agent } i] \geq 0.01,
\end{align*}
which concludes our proof of the lemma.
\end{proof}

With \Cref{lem:agent-level-ef-lb} in hand, we can now prove \Cref{thm:agent-level-ef-lb} in a similar manner as \Cref{thm:agent-level-prop-lb}.

\begin{proof}[Proof of \Cref{thm:agent-level-ef-lb}]
Let $\zeta = 0.01$ and let $\cM$ be any agent-level $\eps$-DP algorithm. 
Consider the input utility functions $\bu' = (u'_1, \dots, u'_n)$ where the utility functions are all-zero, and consider the distribution $\cM(\bu')$.
Notice that, for any allocation $A$, it holds that $\Pr_{i \in N}\left[\rank(i; A) \geq \frac{n-1}{2}\right] \geq 1/2$. 
As a result, we have that $\Pr_{i \in N, A \sim \cM(\bu')}[\rank(i; A) \geq \frac{n-1}{2}] \geq 1/2$. 
Thus, there exists $i^* \in N$ with the property that $\Pr_{A \sim \cM(\bu')}[\rank(i^*; A) \geq \frac{n-1}{2}] \geq 1/2$.

Applying \Cref{lem:agent-level-ef-lb}, we get
\begin{align*}
&\Pr_{\bu, A \sim \cM(\bu')}[A \text{ is not EF}c \text{ for agent } i^*] \\
&\geq \Pr_{A \sim \cM(\bu')}\left[\rank(i^*; A) \geq \frac{n-1}{2}\right] \\
&\cdot\Pr_{\bu, A \sim \cM(\bu')}\bigg[A \text{ is not EF}c \text{ for } i^* \,\bigg\vert \rank(i^*; A) \geq \frac{n-1}{2}\bigg] \\
&\geq \frac{1}{2} \cdot \frac{1}{100} = \frac{1}{200}.
\end{align*}
Hence, there exists $u^*_{i^*}$ such that $$\Pr_{A \sim \cM(\bu')}[A \text{ is not EF}c \text{ for agent } i^* \text{ w.r.t. } u^*_{i^*}] \geq \frac{1}{200}.$$ 

Now, let $\bu^*$ be the input utility such that $u^*_i$ is all-zero for each $i \ne i^*$ while $u^*_{i^*}$ is as above. 
Notice that $\bu^*$ is adjacent to $\bu'$ under agent-level adjacency. 
Thus, applying the $\eps$-DP guarantee of $\cM$, we get
\begin{align*}
&\Pr_{A \sim \cM(\bu^*)}[A \text{ is not EF}c \text{ for agent } i^* \text{ w.r.t. } u^*_{i^*}] \\
&\geq e^{-\eps} \cdot \Pr_{A \sim \cM(\bu')}[A \text{ is not EF}c \text{ for agent } i^* \text{ w.r.t. } u^*_{i^*}] \\
&\geq \frac{e^{-\eps}}{200}.
\end{align*}
This completes the proof.
\end{proof}

\section{(Agent $\times$ Item)-Level DP}
\label{sec:agent-item-level}

In this section, we turn our attention to (agent $\times$ item)-level DP, which is a more relaxed notion than agent-level DP.
We explore both the possibilities (\Cref{sec:agent-item-algo}) and limits (\Cref{sec:agent-item-lower}) of private algorithms with respect to this notion.

\subsection{Algorithms}
\label{sec:agent-item-algo}

In contrast to agent-level DP, we will show that $O_{\eps, n}(\log m)$ upper bounds can be attained in the (agent $\times$ item)-level DP setting. 
Before we do so, let us first explain why straightforward approaches do not work.
To this end, assume that utilities are additive and $u_i(j) \in [0, 1]$ for all $i \in N, j \in M$. 
One may want to estimate $u_i(S)$ for each $S$ using the Laplace mechanism. 
While the Laplace mechanism guarantees that the estimate has an expected additive error of $O(1/\eps)$, this is not useful for obtaining approximate envy-freeness or proportionality guarantees: it is possible that (almost) every good yields utility much less than $1$. 
In this case, additive errors do not translate to any non-trivial EF$c$ or PROP$c$ guarantees. 
We will therefore develop different---and more robust---comparison methods, which ultimately allow us to overcome the aforementioned issue.

\subsubsection{Approximate Envy-Freeness}

Our main algorithmic result for approximate envy-freeness is stated below. 
Note that this result holds even for non-additive utility functions.

\begin{theorem} \label{thm:alg-ef-em}
For any $\eps > 0$ and $\beta \in (0, 1]$, there exists an (agent $\times$ item)-level $\eps$-DP algorithm that, for any input monotone utility functions, outputs a connected EF$c$ allocation with probability at least $1 - \beta$, where $c = O\left(1 + \frac{n \log(mn) + \log(1/\beta)}{\eps}\right)$.
\end{theorem}

The high-level idea of our algorithm is to apply the exponential mechanism (\Cref{thm:EM}) to select an allocation among the $\leq (mn)^n$ \emph{connected} allocations. 
This gives us an ``error'' in the score 
of $O_\eps(\log((mn)^n)) = O_\eps(n \cdot \log (mn))$.
The question is how to set up the score so that (i) it has low sensitivity and (ii) such an error translates to an approximate envy-freeness guarantee. 
Our insight is to define the score based on the following modified utility function that ``removes'' a certain number of most valuable items.

\begin{definition}
\label{def:remove-k}
For $\bu = (u_1, \dots, u_n)$ and $k \in \N \cup \{0\}$, we define $\bu^{-k} = (u^{-k}_1, \dots, u^{-k}_n)$ by
\begin{align*}
u^{-k}_i(S) := \min_{T \subseteq M, |T| \leq k} u_i(S \setminus T) & &\forall i \in N, S \subseteq M.
\end{align*}
\end{definition}

It is clear that $\bu^{-k}$ inherits the monotonicity of $\bu$.
We next list some simple but useful properties of such utility functions. 
The first property, whose proof is trivial, relates \Cref{def:remove-k} to approximate envy-freeness.

\begin{observation} \label{obs:score-to-ef}
Let $k, d \in \N \cup \{0\}$. Any allocation that is EF$d$ with respect to $\bu^{-k}$ is EF$(d+k)$ with respect to $\bu$.
\end{observation}

The second property is that the $d, k$ values are robust with respect to (agent $\times$ item)-level adjacency. 

\begin{lemma} \label{lem:stable}
Let $k \in \N, d \in \N \cup \{0\}$, and $\bu, \tbu$ be any two (agent $\times$ item)-level adjacent inputs. If an allocation $A$ is EF$d$ with respect to $\bu^{-k}$, then it is EF$(d+2)$ with respect to $\tbu^{-(k - 1)}$.
\end{lemma}

\begin{proof}
By definition of (agent $\times$ item)-level adjacency, there exist $i^* \in N, j^* \in M$ such that
\begin{itemize}
\item $u_i = \tu_i$ for all $i \in N \setminus \{i^*\}$, and
\item $u_i(S) = \tu_i(S)$ for all $S \subseteq M \setminus \{j^*\}$.
\end{itemize}

Consider any $i, i' \in N$.
Since $A$ is EF$d$ with respect to $\bu^{-k}$, we have
\begin{align} \label{eq:ef-condition}
u^{-k}_i(A_i) &\geq \min_{S \subseteq M, |S| \le d} u^{-k}_i(A_{i'} \setminus S) \nonumber \\
&= \min_{T \subseteq M, |T| \le d + k} u_i(A_{i'} \setminus T).
\end{align}
Furthermore, we have
\begin{align}
u^{-k}_i(A_i) &= \min_{T \subseteq M, |T| \leq k} u_i(A_i \setminus T) \nonumber \\
&\leq \min_{T \subseteq M, |T| \leq k - 1} u_i(A_i \setminus (T \cup \{j^*\})) \nonumber \\
&= \min_{T \subseteq M, |T| \leq k - 1} \tu_i(A_i \setminus (T \cup \{j^*\})) \nonumber \\
&\leq  \min_{T \subseteq M, |T| \leq k - 1} \tu_i(A_i \setminus T)  = \tu^{-(k - 1)}_i(A_i). \label{eq:remove-one-stable-selected}
\end{align}
Moreover,
\begin{align}
\min_{T \subseteq M, |T| \le d + k} &u_i(A_{i'} \setminus T) \nonumber \\
&\geq \min_{T \subseteq M, |T| \le d + k} u_i(A_{i'} \setminus (T \cup \{j^*\})) \nonumber \\
&= \min_{T \subseteq M, |T| \le d + k} \tu_i(A_{i'} \setminus (T \cup \{j^*\})) \nonumber \\
&\geq \min_{T \subseteq M, |T| \le d + k + 1} \tu_i(A_{i'} \setminus T) \nonumber \\
&= \min_{S \subseteq M, |S| \le d + 2} \tu^{-(k - 1)}_i(A_{i'} \setminus S). \label{eq:remove-one-stable-not-selected}
\end{align}
Thus, we can conclude that
\begin{align*}
\tu^{-(k - 1)}_i(A_i) 
&\overset{\eqref{eq:remove-one-stable-selected}}{\geq} u^{-k}_i(A_i) \\
&
\overset{\eqref{eq:ef-condition}}{\geq} \min_{T \subseteq M, |T| \le d + k} u_i(A_{i'} \setminus T) \\
&\overset{\eqref{eq:remove-one-stable-not-selected}}{\geq} \min_{S \subseteq M, |S| \le d + 2} \tu^{-(k - 1)}_i(A_{i'} \setminus S).
\end{align*}
It follows that $A$ is EF$(d+2)$ with respect to $\tbu^{-(k - 1)}$.
\end{proof}

We can now define our scoring function.

\begin{definition} \label{def:score-em} 
For an allocation~$A$, $\bu = (u_1, \dots, u_n)$, and $g \in \N$, define
$$\scr^g_A(\bu) := - \min\{t \in [g] \mid A \text{ is EF}(2t) \text{ w.r.t. } \bu^{-(g - t)}\}.$$
We let $\scr^g_A(\bu) = -g$ if the set above is empty.
\end{definition}

The following lemma shows that this scoring function has low sensitivity.

\begin{lemma} \label{lem:sensitivity-score}
For any allocation $A$ and $g \in \N$,
$\Delta(\scr^g_A) \leq 1$.
\end{lemma}

\begin{proof}
Let $\bu, \tbu$ be any pair of (agent $\times$ item)-level adjacent inputs. 
Assume without loss of generality that $\scr^g_A(\bu) \geq \scr^g_A(\tbu)$. 
Let $t^* := -\scr^g_A(\bu)$. 

If $t^* = g$, then $\scr^g_A(\tbu) \le \scr^g_A(\bu) = -g$, so $\scr^g_A(\tbu) = -g$.
Otherwise, $t^* \leq g - 1$, and $A$ is EF$(2t^*)$ with respect to $\bu^{-(g - t^*)}$.
\Cref{lem:stable} ensures that $A$ is EF$(2(t^* + 1))$ with respect to $\tbu^{-(g - t^* - 1)}$. 
Thus, $\scr^g_A(\tbu) \geq -(t^* + 1)$. 
Since $-t^* = \scr^g_A(\bu) \geq \scr^g_A(\tbu)$, we have $|\scr^g_A(\bu) - \scr^g_A(\tbu)| \leq 1$ in both cases, completing the proof.
\end{proof}

With all the ingredients ready, we can prove \Cref{thm:alg-ef-em} by applying the exponential mechanism with appropriate parameters (see \Cref{alg:scoring}).

\begin{algorithm}[tb]
\caption{(Agent$\times$item)-level $\varepsilon$-DP algorithm for EF$c$}
\label{alg:scoring}
\textbf{Parameter}: $\eps > 0, \beta\in(0,1]$
\begin{algorithmic}[1] 
\STATE $g \leftarrow 4\left\lceil1 + \frac{\log((mn)^n / \beta)}{\eps}\right\rceil$
\RETURN the allocation output by the $\eps$-DP exponential mechanism using the scoring function $\scr^g_A$ (\Cref{def:score-em}) with the candidate set $\cPc(m, n)$ (\Cref{sec:fair-division}).
\end{algorithmic}
\end{algorithm}

\begin{proof}[Proof of \Cref{thm:alg-ef-em}]
Let $g = 4\left\lceil1 + \frac{\log((mn)^n / \beta)}{\eps}\right\rceil$. 
We run the exponential mechanism using the scoring function $\scr^g_A$ with the candidate set $\cPc(m, n)$. 
By \Cref{thm:EM}, this is an $\eps$-DP algorithm that, for each $\bu$, with probability at least $1 - \beta$, outputs an allocation $A^*$ such that
\begin{align*}
\scr^g_{A^*}(\bu) \geq \max_{A \in \cPc(m,n)} \scr^g_{A}(\bu) - \frac{2\log\left(\frac{|\cPc(m, n)|}{\beta}\right)}{\eps}.
\end{align*}

Fix any $\bu$, and define $A^*$ as above.
By \Cref{thm:connected-ef2}, there exists a connected allocation $A^{\text{EF2}}$ that is EF2 with respect to $\bu^{-(g - 1)}$. 
This means that $\scr^g_{A^{\text{EF2}}} = -1$.
Furthermore, we have\footnote{Indeed, from the set $M = \{1,2,\dots,m\}$, we can allocate one ``block'' of items at a time starting from items with lower indices.
There are at most $m$ possibilities for the size of the next block, this block can be allocated to one of the (at most $n$) remaining agents, and we allocate $n$ blocks in total, hence the bound $(mn)^n$.} $|\cPc(m, n)| \leq (mn)^n$. 
Plugging these into the inequality above, we get
\begin{align*}
\scr^g_{A^*}(\bu) \geq -1 - \frac{2\log((mn)^n / \beta)}{\eps} \geq -\frac{g}{2},
\end{align*}
where the latter inequality follows from our choice of $g$. 
Hence, $A^*$ is EF$(2c)$ with respect to $\bu^{-(g-c)}$ for some $c \le g/2$. 
Invoking \Cref{obs:score-to-ef}, we find that $A^*$ is EF$(g+c)$, and therefore EF$(3g/2)$, with respect to $\bu$.
This concludes our proof. 
\end{proof}

\subsubsection{Approximate Proportionality}

Next, we present an improved result for approximate proportionality, where the dependence on $n$ is reduced to $O(\log n)$.

\begin{theorem} \label{thm:prop-moving-knife}
For any $\eps > 0$ and $\beta \in (0, 1]$, there exists an (agent $\times$ item)-level $\eps$-DP algorithm that, for any input additive utility functions, outputs a connected PROP$c$ allocation with probability at least $1 - \beta$, where $c = O\left(\log n + \frac{\log(mn/\beta)}{\eps}\right)$.
\end{theorem}

Our algorithm is based on the well-known ``moving-knife'' procedure from cake cutting \citep{DubinsSp61}.
A natural way to implement this idea in our setting is to place the items on a line, put a knife at the left end, and move it rightwards until some agent values the subset of items to the left of the knife at least $1/n$ of her value for the whole set of items.
We give this subset to this agent, and proceed similarly with the remaining agents and items. 
To make this procedure DP, we can replace the check of whether each agent receives sufficiently high value with the SVT algorithm (\Cref{thm:svt}), where the usual utility is modified similarly to \Cref{def:score-em} to achieve low sensitivity. 
While this approach is feasible, it does not establish the bound we want: since the last agent has to participate in $n$ ``rounds'' of this protocol, the basic composition theorem (\Cref{thm:basic-comp}) implies that we can only allot a privacy budget of $\eps / n$ in each round. 
This results in a guarantee of the form $c = O(\log m / (\eps/n)) = O((n \log m) / \eps)$, which does not distinctly improve upon the guarantee in \Cref{thm:alg-ef-em}.

To overcome this issue, notice that instead of targeting a single agent, we can continue moving our knife until at least $n/2$ agents value the subset of items to the left of the knife at least half of the entire set.\footnote{\citet{EvenPa84} used a similar idea in cake cutting.}
This allows us to recurse on both sides, thereby reducing the number of rounds to $\log n$. 
Hence, we may allot a privacy budget of $\eps / \log n$ in each round.
Unfortunately, this only results in a bound of the form $c = O(\log m / (\eps/ \log n)) = O((\log n \log m) / \eps)$, which is still worse than what we claim in \Cref{thm:prop-moving-knife}.

Our last observation is that we can afford to make more mistakes in earlier rounds: for example, in the first round, we would be fine with making an ``error'' of roughly $O(n)$ in the knife position because the subsets on both sides will be subdivided to $\Omega(n)$ parts later. 
As a result, our strategy is to allot less privacy budget in earlier rounds and more in later rounds.
By letting the privacy budgets form a geometric sequence, we can achieve our claimed $O(\log (mn) / \eps)$ bound.

\begin{proof}[Proof of \Cref{thm:prop-moving-knife}]
The proof follows the overview outlined above.
The statement holds trivially if $\beta = 1$, so assume that $\beta < 1$.
For convenience, given any positive integers $\ell \leq r$, we write $[\ell, r]$ as a shorthand for $\{\ell,\ell+1, \dots, r\}$.

Let $\eps_1, \dots, \eps_{\lceil \log_2 n \rceil}, g_1, \dots, g_{\lceil \log_2 n \rceil}$ be defined as $\eps_b = \frac{\eps}{2 \cdot (1.5^b)}$ and $g_b = 8\lceil \upsilon \cdot \log(mn/\beta) / \eps_b \rceil$, where $\upsilon$ is defined as in \Cref{thm:svt}. 
Our algorithm, \textsc{DPMovingKnife}, is presented as \Cref{alg:moving-knife}. 
The final algorithm is an instantiation of \Cref{alg:moving-knife} with $\cI = N$ and $(\ell, r) = (1, m)$.

\begin{algorithm}[tb]
\caption{\textsc{DPMovingKnife}}
\label{alg:moving-knife}
\textbf{Input}: A set $\cI \subseteq N$ of agents, a set $[\ell, r] \subseteq M$ of items\\
\textbf{Parameter}: $\eps_1, \dots, \eps_{\lceil \log n\rceil} > 0, g_1, \dots, g_{\lceil \log n \rceil} \in \N$\\
\textbf{Output}: A partial allocation $(A_i)_{i \in \cI}$ of $[\ell, r]$

\begin{algorithmic}[1] 
\IF{$|\cI| = 1$}
\RETURN $(A_i = [\ell, r])$ for $i\in\cI$
\ENDIF
\STATE $b \leftarrow \lceil \log_2 |\cI| \rceil$
\STATE $n^R \leftarrow \lfloor |\cI| / 2 \rfloor$
\STATE $n^L \leftarrow |\cI| - n^R$
\FOR{each agent $i \in \cI$}
\FOR{each $h \in [\ell, r]$}
\STATE 
\begin{align*}
T^{i, \ell, r}_h(\bu) \gets \bigg\{&t \in [g_b] \,\bigg\vert\, \frac{1}{n^L} \cdot u^{-(g_b + t)}_i([\ell, h]) \\
&\geq \frac{1}{n^R} \cdot u^{-(g_b - t)}_i([h + 1, r])\bigg\}
\end{align*}
\STATE
\begin{align*}
f^{i, \ell, r}_h(\bu) \gets
\begin{cases}
\max(T^{i, \ell, r}_h(\bu)) & \text{ if } T^{i, \ell, r}_h(\bu) \ne \emptyset; \\
0 &\text{ otherwise}
\end{cases}
\end{align*}
\ENDFOR
\STATE $h_i \gets$ output from the $\eps_b$-DP SVT algorithm (\Cref{thm:svt}) with $\tau = g_b/2$ on $f^{i, \ell, r}_\ell, \dots, f^{i, \ell, r}_r$
\ENDFOR
\STATE $h_{z_1} \leq h_{z_2} \leq \cdots \leq h_{z_{|\cI|}} \gets$ sorted list of $h_i$'s
\STATE $A^L \gets \textsc{DPMovingKnife}(\{z_1, \dots, z_{n^L}\}, \{\ell, \dots, h_{z_{n^L}}\})$
\STATE $A^R \gets \textsc{DPMovingKnife}(\{z_{n^L + 1}, \dots, z_{|\cI|}\}, \{h_{z_{n^L}} + 1, \dots, r\})$
\RETURN the allocation resulting from combining $A^L$ and~$A^R$
\end{algorithmic}
\end{algorithm}

Before we prove the utility and privacy guarantees of our algorithm, let us prove the following auxiliary lemma which states that the sensitivity of $f^{i, \ell, r}_h$ is small. The proof is similar to that of \Cref{lem:sensitivity-score}.

\renewcommand\qedsymbol{$\vartriangleleft$} 

\begin{lemma} \label{lem:sensitivity-f}
For all $\ell, r, h \in M$ and $i \in N$, it holds that $\Delta(f^{i, \ell, r}_h) \leq 1$.
\end{lemma}

\begin{proof}
Let $\bu, \tbu$ be any pair of (agent $\times$ item)-level adjacent inputs. 
Assume without loss of generality that $f^{i, \ell, r}_h(\bu) \geq f^{i, \ell, r}_h(\tbu)$. 
Let $t^* := f^{i, \ell, r}_h(\bu)$. 

If $t^* = 0$, then $f^{i, \ell, r}_h(\tbu)$ must also be equal to 0.

Otherwise, $t^* > 0$, and we have
\begin{align} \label{eq:f-inequality}
\frac{1}{n^L} \cdot u^{-(g_b + t^*)}_i([\ell, h])
\geq \frac{1}{n^R} \cdot u^{-(g_b - t^*)}_i([h + 1, r]).
\end{align}
From this, we can derive
\begin{align*}
\frac{1}{n^L} \cdot \tu^{-(g_b + t^* - 1)}_i([\ell, h]) 
&\geq \frac{1}{n^L} \cdot u^{-(g_b + t^*)}_i([\ell, h]) \\
&\overset{\eqref{eq:f-inequality}}{\geq} \frac{1}{n^R} \cdot u^{-(g_b - t^*)}_i([h + 1, r]) \\
&\geq \frac{1}{n^R} \cdot \tu^{-(g_b - t^* + 1)}_i([h + 1, r]),
\end{align*}
where the first and last inequalities follow from the fact that $\bu, \tbu$ are (agent $\times$ item)-level adjacent inputs. Thus, it must be that $f^{i, \ell, r}_h(\tbu) \geq t^* - 1$. 

It follows that $|f^{i, \ell, r}_h(u) - f^{i, \ell, r}_h(\tbu)| \leq 1$ in both cases, completing the proof.
\end{proof}

We are now ready to establish the privacy and utility guarantees of the algorithm, starting with the former.

\paragraph{Privacy Analysis.}
We will prove by induction that $\textsc{DPMovingKnife}(N, M)$ is $\eps$-DP.
Specifically, we claim that $\textsc{DPMovingKnife}(\cI, [\ell, r])$ is $\left(\sum_{b=1}^{\lceil \log_2 |\cI| \rceil} \eps_b\right)$-DP for any $\cI \in N$ and $[\ell, r] \subseteq M$. 
We prove this claim by (strong) induction on the size of $\cI$. 
The base case $|\cI| = 1$ is trivial.

For $|\cI| > 1$, let us divide the algorithm into two stages: (i) computation of $h_{z_1}, \dots, h_{z_{|\cI|}}$ and (ii) computation of $A^L$ and $A^R$ given $h_{z_1}, \dots, h_{z_{|\cI|}}$. 
The privacy for each stage can be analyzed as follows:
\begin{itemize}
\item \textbf{Stage (i).} 
From \Cref{lem:sensitivity-f} and \Cref{thm:svt}, each application of SVT is $(\eps_{\lceil \log_2 |\cI| \rceil})$-DP. 
Moreover, since SVT is applied on each $u_i$ separately, parallel composition (\Cref{thm:parallel-comp}) implies\footnote{More specifically, we may apply \Cref{thm:parallel-comp} with $\Gamma: \cX \to (\cX)^{\cI}$ where $\Gamma(\bu)_i = u_i$, and $\cM$ being the SVT algorithm.} that the entire computation of $(h_i)_{i \in \cI}$ is also $(\eps_{\lceil \log_2 |\cI| \rceil})$-DP. 
Finally, $h_{z_1}, \dots, h_{z_{|\cI|}}$ is simply a post-processing of $(h_i)_{i \in \cI}$, so \Cref{obs:post-processing} ensures that the entire Stage (i) is $(\eps_{\lceil \log_2 |\cI| \rceil})$-DP.
\item \textbf{Stage (ii).} 
The inductive hypothesis asserts that each recursive call to \textsc{DPMovingKnife} is $\left(\sum_{b=1}^{\lceil \log_2|\cI| \rceil - 1} \eps_b\right)$-DP. 
Furthermore, since each $u_i$ is used in only one recursive call, we can apply\footnote{More specifically, we let $\Gamma: \cX \to (\cX)^2$ where $\Gamma(\bu)_1 = (u_{z_1}, \dots, u_{z_{n^L}})$ and $\Gamma(\bu)_2 = (u_{z_{n^L + 1}}, \dots, u_{z_{|\cI|}})$, and $\cM$ being the \textsc{DPMovingKnife} algorithm.} parallel composition (\Cref{thm:parallel-comp}), which ensures that the entire Stage (ii) is also $\left(\sum_{b=1}^{\lceil \log_2|\cI| \rceil - 1} \eps_b\right)$-DP.
\end{itemize}
Therefore, applying basic composition (\Cref{thm:basic-comp}) across the two stages yields that the entire algorithm is $\left(\sum_{b=1}^{\lceil \log_2 |\cI| \rceil} \eps_b\right)$-DP.
Since $\sum_{b=1}^{\lceil \log_2 n \rceil} \eps_b \leq \sum_{b=1}^{\infty} \eps_b = \eps$, it follows that the entire algorithm when called with $\cI = N$ is $\eps$-DP, as desired.

\paragraph{Utility Analysis.}
We next analyze the utility of the algorithm. 
To this end, for each agent $i$, let $\ell^i_1, r^i_1, \dots, \ell^i_{w_i}, r^i_{w_i}$ be the values of $\ell, r$ with which the algorithm is invoked for a set $\cI$ containing $i$. 
Note that $\ell^i_1 = 1$, $r^i_1 = m$, and $w_i \leq \lceil \log_2 n \rceil + 1$. 
Similarly, let $h^i_1, \dots, h^i_{w_i - 1}$ denote the values of $h_i$ output by SVT for agent $i$ in each of the calls except the last one, and $n^i_1, \dots, n^i_{w_i}$ denote the sizes of $\cI$ in each invocation. (Note that we must have $n^i_{w_i} = 1$.)
Furthermore, for each $q \in [w_i - 1]$, we let $R^i_q := [\ell^i_q, r^i_q] \setminus [\ell^i_{q + 1}, r^i_{q + 1}]$ denote the set of items removed in the $q$-th iteration, and let $u^{\topk}_i(R^i_q) := \max_{S \subseteq R^i_q, |S| \leq k} u_i(S)$ for any non-negative integer~$k$.

Consider any $i \in N$ and $q \in [w_i - 1]$. 
Let $b_q = \lceil \log_2 n_q^i \rceil$ Notice that $f^{i, \ell, r}_r(\bu) = g_{b_q} \geq g_{b_q}/2$. 
Thus, by the guarantee of SVT (\Cref{thm:svt}), we have that with probability at least $1 - \beta/n^2$, the following holds for each $i \in N$ and $q \in [w_i - 1]$:
\begin{align} \label{eq:svt-above}
f^{i,\ell,r}_{h^i_q}(\bu) 
&\geq \frac{g_{b_q}}{2} - \upsilon \cdot \frac{\log(mn^2/\beta)}{\eps_{b_q}}  \nonumber \\
&> \frac{g_{b_q}}{2} - 2\upsilon \cdot \frac{\log(mn/\beta)}{\eps_{b_q}} \geq \frac{g_{b_q}}{4}
\end{align} 
and
\begin{align} \label{eq:svt-below}
f^{i,\ell,r}_{h'}(\bu) 
&\leq \frac{g_{b_q}}{2} + \upsilon \cdot \frac{\log(mn^2/\beta)}{\eps_{b_q}} \nonumber \\
&< \frac{g_{b_q}}{2} + 2\upsilon \cdot \frac{\log(mn/\beta)}{\eps_{b_q}}  \leq \frac{3g_{b_q}}{4}
\end{align}
for all $h' < h^i_q$.
Therefore, the union bound implies that, with probability at least $1 - \beta$, both \eqref{eq:svt-above} and \eqref{eq:svt-below} hold for all $i\in N$ and $q\in [w_i - 1]$ simultaneously.
We will show that, when this occurs, the output allocation is PROP$c$ for the value of $c$ in the theorem statement.

First, let us fix a pair $i \in N$ and $q \in [w_i - 1]$, and consider the $q$-th time \textsc{DPMovingKnife} is called with $i \in \cI$. 
We will show that
\begin{align} \label{eq:one-step-moving-knife}
&u_i([\ell^i_{q + 1}, r^i_{q + 1}])\nonumber\\
&\geq \frac{n^i_{q+1}}{n^i_q} \left(u_i([\ell^i_q, r^i_q]) - u_i^{\topv(2g_{b_q})}(R^i_q)\right).
\end{align}
To do so, consider two cases, based on whether $i$ is recursed on the left subinstance\footnote{Note that all unspecified variables are for the run of the algorithm as specified above.
} (i.e., $i \in \{z_1, \dots, z_{n^L}\}$) or on the right subinstance (i.e., $i \in \{z_{n^L+1}, \dots, z_{n^i_q}\}$). \begin{itemize}

\item \underline{Case I}: $i \in \{z_1, \dots, z_{n^L}\}$. In this case, we have
\begin{align*}
u_i([\ell^i_{q + 1}, r^i_{q + 1}])
&= u_i([\ell, h_{z_{n^L}}]) \\
&\geq u_i([\ell, h^i_q]) \\
&\overset{\eqref{eq:svt-above}}{\geq}  \frac{n^L}{n^R} \cdot u_i^{-3g_{b_q}/4}([h^i_q + 1, r]) \\
&\geq \frac{n^L}{n^R} \cdot u_i^{-3g_{b_q}/4}([h_{z_{n^L}} + 1, r]),
\end{align*}
where for the second inequality we apply \eqref{eq:svt-above} on the definition of $f^{i,\ell,r}_{h^i_q}$ with $t \ge g_{b_q}/4$.
Therefore, we have
\begin{align*}
&u_i([\ell^i_{q + 1}, r^i_{q + 1}])\\
&\geq \frac{n^L}{|\cI|} u_i([\ell, h_{z_{n^L}}]) \\
&\qquad+ \frac{n^R}{|\cI|} \cdot \frac{n^L}{n^R} \cdot u_i^{-3g_{b_q}/4}([h_{z_{n^L}} + 1, r]) \\
&= \frac{n^L}{|\cI|} \bigg(u_i([\ell, h_{z_{n^L}}]) + u_i^{-3g_{b_q}/4}([h_{z_{n^L}}+ 1, r])\bigg) \\
&= \frac{n^L}{|\cI|} \bigg(u_i([\ell, r]) - u_i^{\topv(3g_{b_q}/4)}([h_{z_{n^L}} + 1, r])\bigg) \\
&= \frac{n^i_{q+1}}{n^i_q} \left(u_i([\ell^i_q, r^i_q]) - u_i^{\topv(3g_{b_q}/4)}(R^i_q)\right) \\
&\geq \frac{n^i_{q+1}}{n^i_q} \left(u_i([\ell^i_q, r^i_q]) - u_i^{\topv(2g_{b_q})}(R^i_q)\right).
\end{align*}

\item \underline{Case II}: $i \in \{z_{n^L+1}, \dots, z_{n^i_q}\}$. In this case, we have
\begin{align*}
u_i([\ell^i_{q + 1}, r^i_{q + 1}])
&= u_i([h_{z_{n^L}} + 1, r]) \\
&\geq u_i([h^i_q + 1, r]) \\
&\geq u^{-1}_i([h^i_q, r]) \\
&\overset{\eqref{eq:svt-below}}{>} \frac{n^R}{n^L} \cdot u_i^{-7g_{b_q}/4}([\ell, h^i_q - 1]) \\
&\geq \frac{n^R}{n^L} \cdot u_i^{-(7g_{b_q}/4+1)}([\ell, h^i_q]) \\
&\geq \frac{n^R}{n^L} \cdot u_i^{-(7g_{b_q}/4 + 1)}([\ell, h_{z_{n^L}}]) \\
&\geq \frac{n^R}{n^L} \cdot u_i^{-2g_{b_q}}([\ell, h_{z_{n^L}}]),
\end{align*}
where the third inequality follows from applying \eqref{eq:svt-below} with $h' = h^i_q - 1$ and $t = 3 g_{b_q}/4$ in the definition of $f^{i,\ell,r}_{h^i_q}$.

Therefore, we have
\begin{align*}
&u_i([\ell^i_{q + 1}, r^i_{q + 1}])\\
&\geq \frac{n^L}{|\cI|} \cdot \frac{n^R}{n^L} \cdot u_i^{-2g_{b_q}}([\ell, h_{z_{n^L}}]) \\
&\qquad+ \frac{n^R}{|\cI|} \cdot u_i([h_{z_{n^L}} + 1, r]) \\
&= \frac{n^R}{|\cI|} \bigg(u_i^{-2g_{b_q}}([\ell, h_{z_{n^L}}]) + u_i([h_{z_{n^L}} + 1, r]) \bigg) \\
&= \frac{n^R}{|\cI|} \left(u_i([\ell, r]) - u_i^{\topv(2g_{b_q})}([\ell, h_{z_{n^L}}])\right) \\
&= \frac{n^i_{q+1}}{n^i_q} \left(u_i([\ell^i_q, r^i_q]) - u_i^{\topv(2g_{b_q})}(R^i_q)\right).
\end{align*}
\end{itemize}

Thus, \eqref{eq:one-step-moving-knife} holds in both cases.

Let $A = (A_1,\dots,A_n)$ be the allocation returned by the algorithm.
For each $i\in N$, by repeatedly applying \eqref{eq:one-step-moving-knife}, we get
\begin{align*}
&u_i(A_i) \\
&= u_i([\ell^i_{w_i}, r^i_{w_i}]) \\
&\overset{\eqref{eq:one-step-moving-knife}}{\geq} \frac{1}{n^i_{w_i - 1}} \bigg(u_i([\ell^i_{w_i-1}, r^i_{w_i - 1}]) - u_i^{\topv(2g_{b_{w_i - 1}})}(R^i_{w_i - 1})\bigg) \\
&\overset{\eqref{eq:one-step-moving-knife}}{\geq} \frac{1}{n^i_{w_i - 2}} \bigg(u_i([\ell^i_{w_i-2}, r^i_{w_i - 2}]) - u_i^{\topv(2g_{b_{w_i - 2}})}(R^i_{w_i - 2})\bigg) \\
&\qquad- \frac{1}{n^i_{w_i - 1}} \cdot u_i^{\topv(2g_{b_{w_i - 1}})}(R^i_{w_i - 1}) \\
&\overset{\eqref{eq:one-step-moving-knife}}{\geq} \cdots \\
&\overset{\eqref{eq:one-step-moving-knife}}{\geq} \frac{1}{n^i_1} \cdot u_i([\ell^i_1, r^i_1]) - \sum_{q \in [w_i - 1]} \frac{1}{n^i_q} \cdot u_i^{\topv(2g_{b_q})}(R^i_q)  \\
&= \frac{1}{n} \cdot u_i(M) - \sum_{q \in [w_i - 1]} \frac{1}{n^i_q} \cdot u_i^{\topv(2g_{b_q})}(R^i_q) \\
&\geq \frac{1}{n} \cdot u_i(M) - \sum_{q \in [w_i - 1]} u_i^{\topv(\lceil 2g_{b_q}/n^i_q\rceil)}(R^i_q) \\
&\geq \frac{1}{n} \cdot u_i(M) - u_i^{\topv\left(\sum_{q \in [w_i - 1]} \lceil 2g_{b_q}/n^i_q\rceil\right)}\left(\bigcup_{q \in [w_i 
- 1]} R^i_q\right) \\
&= \frac{1}{n} \cdot u_i(M) - u_i^{\topv\left(\sum_{q \in [w_i - 1]} \lceil 2g_{b_q}/n^i_q\rceil\right)}\left(M \setminus A_i\right),
\end{align*}
where the last inequality follows from the fact that $R^i_1, \dots, R^i_{w_i - 1}$ are disjoint.

Now, we can bound the term $\sum_{q \in [w_i - 1]} \lceil 2g_{b_q}/n^i_q \rceil$ as follows:
\begin{align*}
&\sum_{q \in [w_i - 1]} \left\lceil \frac{2g_{b_q}}{n^i_q} \right\rceil \\
&\leq \sum_{q \in [w_i - 1]} \left(\frac{2g_{b_q}}{n^i_q} + 1\right) \\
&\leq (w_i - 1) + \sum_{q \in [w_i - 1]} O\left(\frac{\log(mn/\beta)/\eps_{b_q} + 1}{n^i_q}\right) \\
&\leq O(\log n) +  \sum_{q \in [w_i - 1]} O\left(\frac{\log(mn/\beta) \cdot 1.5^{b_q}/ \eps + 1}{2^{b_q}}\right) \\
&= O(\log n) + \sum_{q \in [w_i - 1]} O\left(\log(mn/\beta) / \eps\right) \cdot (0.75)^{b_q} \\
&\leq O(\log n) + O\left(\log(mn/\beta) / \eps)\right),
\end{align*}
where the last inequality follows from the fact that the $b_q$'s are different for different values of $q \in [w_i - 1]$.

Putting the previous two bounds together, we can conclude that the allocation~$A$ is PROP$c$ for $c = O(\log n + \log(mn/\beta) / \eps)$, as desired.
\renewcommand\qedsymbol{$\square$}
\end{proof}

\subsection{Lower Bounds}
\label{sec:agent-item-lower}

Next, we prove lower bounds for (agent $\times$ item)-level DP via the \emph{packing method}~\cite{HardtTa10}.
This involves constructing inputs that are close to one another (with respect to the corresponding adjacency notion) such that the acceptable solutions (i.e., EF$c$ or PROP$c$ allocations) are different for different inputs.
The DP requirement can then be used to rule out the existence of algorithms with strong utility guarantees.
We reiterate that our lower bounds hold only against connected allocations.
In our constructions, we design the utility functions so that each input forces us to pick particular positions to cut in order to get an EF$c$ or PROP$c$ allocation.
We start with the proof for envy-freeness. 

\begin{theorem} \label{thm:lb-agent-item-dp-connected}
There exists $\zeta > 0$ such that, for any $\eps \in (0, 1]$, there is no $\eps$-DP algorithm that, for any input binary additive utility functions, outputs a connected EF$c$ allocation with probability at least $0.5$, where $c = \left\lfloor \zeta \cdot \min\left\{\frac{\log m}{\eps}, \frac{m}{n}, \sqrt{m}\right\} \right\rfloor$.
\end{theorem}

\begin{proof}
Let $\zeta = 0.01$, $c$ be as in the theorem statement, and $T = \lfloor m / (4c + 4) \rfloor$. 
We may assume that $c \geq 1$, as otherwise the theorem holds trivially even without the privacy requirement.
Consider the following utility functions.
\begin{itemize}
\item Let $\bu' = (u'_1, \dots, u'_n)$ denote the binary additive utility functions defined as follows:
\begin{itemize}
\item $u'_1$ and $u'_2$ are all-zero utility functions.
\item For all $i \in \{3, \dots, n\}$ and $j \in M$, let
\begin{align*}
u'_i(j) =
\begin{cases}
1 & \text{ if } j \geq m - (c + 1)(n - 2); \\
0 & \text{ otherwise.}
\end{cases}
\end{align*}
\end{itemize}
\item For every $t \in [T]$, let $\bu^t = (u_1^t, \dots, u_n^t)$ denote the binary additive utility functions defined as follows:
\begin{itemize}
\item $u^t_1$ and $u^t_2$ are defined as follows:
\begin{align*}
u^t_1(j) = u^t_2(j) =
\begin{cases}
1 & \text{ if } \left\lfloor \frac{j - 1}{2c + 1} \right\rfloor = t - 1; \\
0 & \text{ otherwise,}
\end{cases}
\end{align*}
for all $j \in M$.
\item For all $i \in \{3, \dots, n\}$, $u^t_i$ is exactly the same as $u'_i$ defined earlier.
\end{itemize}
\end{itemize}

Suppose for contradiction that there is an $\eps$-DP algorithm $\cM$ that, with probability at least $0.5$, outputs a connected allocation that is EF$c$ for its input utility functions. 
For each $t \in [T]$, let $\cA^t$ denote the set of allocations that are EF$c$ for $\bu^t$. 
The assumption on $\cM$ can be written as
\begin{align} \label{eq:utility-ef-lb}
\Pr[\cM(\bu^t) \in \cA^t] \geq 0.5.
\end{align}

Let $\sim$ denote the (agent $\times$ item)-level adjacency relation. 
One can check that $\bu' \sim_{4c + 2} \bu^t$ for all $t \in [T]$. 
Using this fact together with group DP (\Cref{lem:group-dp}), we have
\begin{align} \label{eq:group-dp-application}
\Pr[\cM(\bu') \in \cA^t] 
&\geq e^{-\eps(4c + 2)} \cdot \Pr[\cM(\bu^t) \in \cA^t] \nonumber\\
&\overset{\eqref{eq:utility-ef-lb}}{\geq} 0.5 \cdot e^{-\eps(4c + 2)}.
\end{align}

\renewcommand\qedsymbol{$\vartriangleleft$} 

\begin{lemma} \label{lem:packing-disjoint}
$\cA^1, \dots, \cA^T$ are disjoint.
\end{lemma}

\begin{proof}
Suppose for contradiction that $\cA^t$ and $\cA^{t'}$ overlap for some $t < t'$. 
Let $A$ be an allocation that is contained in both $\cA^t$ and $\cA^{t'}$. Observe that $A$ must satisfy the following:
\begin{itemize}
\item There must be at least $n - 3$ cuts between item $m - (c + 1)(n - 2)$ and item $m$. (Otherwise, the allocation cannot be EF$c$ for at least one of agents $3, \dots, n$.)
\item There must be a cut between item $(2c + 1)(t - 1) + 1$ and item $(2c + 1)t$. (Otherwise, the allocation cannot be EF$c$ for either agent $1$ or $2$ in $\bu^t$.)
\item There must be a cut between item $(2c + 1)(t' - 1) + 1$ and item $(2c + 1)t'$. (Otherwise, the allocation cannot be EF$c$ for either agent $1$ or $2$ in $\bu^{t'}$.)
\end{itemize}
By our parameter selection, we have
\begin{align*}
&\max\{(2c + 1)t, (2c + 1)t'\} \\
&\leq (2c + 1)T
\leq m/2 
< m - 2cn
\leq m - (c+1)(n-2).
\end{align*}
This implies that the intervals of items discussed above are pairwise disjoint, and therefore the cuts must be distinct, i.e., those $n - 1$ cuts are all the cuts in $A$. 
Now, consider the leftmost interval of $A$. 
The cut for this interval must be to the left of item $(2c + 1)t$. 
Suppose that the interval is assigned to agent~$i$. We consider the following two cases.
\begin{itemize}
\item \underline{Case I}: $i \in \{3, \dots, n\}$. 
In this case, agent $i$ has value~$0$ for her own bundle. 
Furthermore, since there are only \mbox{$n - 3$} cuts between item $m - (c + 1)(n - 2)$ and item $m$, at least one other bundle has at least $c + 1$ items valued $1$ by agent $i$.
As such, $A$ cannot be EF$c$ for agent $i$.
\item \underline{Case II}: $i \in \{1, 2\}$. In $\bu^{t'}$, agent $i$ values her own bundle zero. 
Furthermore, since there is just one cut between item $(2c + 1)(t' - 1) + 1$ and item $(2c + 1)t'$, at least one other bundle has at least $c + 1$ items valued $1$ by agent $i$. 
As such, the allocation cannot be EF$c$ for agent $i$.
\end{itemize} 
In both cases, $A$ cannot be EF$c$ for both $\bu^t$ and $\bu^{t'}$ simultaneously, a contradiction.
This establishes \Cref{lem:packing-disjoint}.
\end{proof}

\Cref{lem:packing-disjoint} implies that
\begin{align*}
1 &\geq \sum_{t \in [T]} \Pr[\cM(\bu') \in \cA^t] \\
&\overset{\eqref{eq:group-dp-application}}{\geq} 0.5 \cdot e^{-\eps(4c + 2)} \cdot T \\
&\geq 0.5 \cdot e^{-6c\eps} \cdot \left\lfloor\frac{m}{8c}\right\rfloor \\ 
&\geq 0.5 \cdot e^{-0.06\log m} \cdot \lfloor 10 \sqrt{m} \rfloor \\ 
&\geq (0.5 / \sqrt{m}) \cdot (5 \sqrt{m}) > 1,
\end{align*}
where the third and fourth inequalities follow from our choice of parameters $c, T$ and the assumption that $c \geq 1$. 
This is a contradiction which establishes \Cref{thm:lb-agent-item-dp-connected}.
\renewcommand\qedsymbol{$\square$}
\end{proof}

For proportionality, we obtain a slightly weaker bound where the $\log m$ term is replaced by $\log(m/n)/n$. 
The proof is similar to that of \Cref{thm:lb-agent-item-dp-connected}.

\begin{theorem} \label{thm:lb-agent-item-dp-connected-prop}
There exists a constant $\zeta > 0$ such that, for any $\eps \in (0, 1]$, there is no $\eps$-DP algorithm that, for any input binary additive utility functions, outputs a connected PROP$c$ allocation with probability at least $0.5$, where $c = \left\lfloor \zeta \cdot \min\left\{\frac{\log (m/n)}{\eps n}, \frac{m}{n}, \sqrt{\frac{m}{n}}\right\} \right\rfloor$.
\end{theorem}

\begin{proof}
Let $\zeta = 0.01$, $c$ be as in the theorem statement, and $T = \lfloor m / (2nc + 2) \rfloor$. 
We may assume that $c \geq 1$, as otherwise the theorem holds trivially even without the privacy requirement.
Consider the following utility functions.
\begin{itemize}
\item Let $\bu' = (u'_1, \dots, u'_n)$ denote the binary additive utility functions defined as follows:
\begin{itemize}
\item $u'_1$ and $u'_2$ are all-zero utility functions.
\item For all $i \in \{3, \dots, n\}$ and $j \in M$, let
\begin{align*}
u'_i(j) =
\begin{cases}
1 & \text{ if } j \geq m - cn; \\
0 & \text{ otherwise.}
\end{cases}
\end{align*}
\end{itemize}
\item For every $t \in [T]$, let $\bu^t = (u_1^t, \dots, u_n^t)$ denote the binary additive utility functions defined as follows:
\begin{itemize}
\item $u^t_1$ and $u^t_2$ are defined as follows:
\begin{align*}
u^t_1(j) = u^t_2(j) =
\begin{cases}
1 & \text{ if } \left\lfloor \frac{j - 1}{nc + 1} \right\rfloor = t - 1; \\
0 & \text{ otherwise,}
\end{cases}
\end{align*}
for all $j \in M$.
\item For all $i \in \{3, \dots, n\}$, $u^t_i$ is exactly the same as $u'_i$ defined earlier.
\end{itemize}
\end{itemize}
Note that every agent values exactly $cn+1$ items, so each agent needs at least one item that she values in order for PROP$c$ to be satisfied.

Suppose for contradiction that there is an $\eps$-DP algorithm $\cM$ that, with probability at least $0.5$, outputs a connected allocation that is PROP$c$ for its input utility functions. 
For each $t \in [T]$, let $\cA^t$ denote the set of allocations that are PROP$c$ for $\bu^t$. 
The assumption on $\cM$ can be written as
\begin{align} \label{eq:utility-prop-lb}
\Pr[\cM(\bu^t) \in \cA^t] \geq 0.5.
\end{align}

Let $\sim$ denote the (agent $\times$ item)-level adjacency relation. 
One can check that $\bu' \sim_{2nc + 2} \bu^t$ for all $t \in [T]$. 
Using this fact together with group DP (\Cref{lem:group-dp}), we have
\begin{align} \label{eq:group-dp-application-prop}
\Pr[\cM(\bu') \in \cA^t] 
&\geq e^{-\eps(2nc + 2)} \cdot \Pr[\cM(\bu^t) \in \cA^t] \nonumber\\
&\overset{\eqref{eq:utility-prop-lb}}{\geq} 0.5 \cdot e^{-\eps(2nc + 2)}.
\end{align}

\renewcommand\qedsymbol{$\vartriangleleft$} 

\begin{lemma} \label{lem:packing-disjoint-prop}
$\cA^1, \dots, \cA^T$ are disjoint.
\end{lemma}

\begin{proof}
Suppose for contradiction that $\cA^t$ and $\cA^{t'}$ overlap for some $t < t'$. 
Let $A$ be an allocation that is contained in both $\cA^t$ and $\cA^{t'}$. Observe that $A$ must satisfy the following:
\begin{itemize}
\item There must be at least $n - 3$ cuts between item $m - cn$ and item $m$. (Otherwise, the allocation cannot be PROP$c$ for at least one of agents $3, \dots, n$.)
\item There must be a cut between item $(nc + 1)(t - 1) + 1$ and item $(nc + 1)t$. (Otherwise, the allocation cannot be PROP$c$ for either agent $1$ or $2$ in $\bu^t$.)
\item There must be a cut between item $(nc + 1)(t' - 1) + 1$ and item $(nc + 1)t'$. (Otherwise, the allocation cannot be PROP$c$ for either agent $1$ or $2$ in $\bu^{t'}$.)
\end{itemize}
By our parameter selection, we have
\begin{align*}
\max\{(nc + 1)t&, (nc + 1)t'\} \\
&\leq (nc + 1)T
\leq m/2 
< m - cn.
\end{align*}
This implies that the intervals of items discussed above are pairwise disjoint, and therefore the cuts must be distinct, i.e., those $n - 1$ cuts are all the cuts in $A$. 
Now, consider the leftmost interval of $A$. 
The cut for this interval must be to the left of item $(nc + 1)t$. 
Suppose that the interval is assigned to agent~$i$. 
We consider the following two cases.
\begin{itemize}
\item \underline{Case I}: $i \in \{3, \dots, n\}$. 
In this case, agent $i$ has value~$0$ for her own bundle, and PROP$c$ is not satisfied.
\item \underline{Case II}: $i \in \{1, 2\}$. In $\bu^{t'}$, agent $i$ has value~$0$ for her own bundle, and PROP$c$ is again not satisfied. 
\end{itemize} 
In both cases, $A$ cannot be PROP$c$ for both $\bu^t$ and $\bu^{t'}$ simultaneously, a contradiction.
This establishes \Cref{lem:packing-disjoint-prop}.
\end{proof}

\Cref{lem:packing-disjoint-prop} implies that
\begin{align*}
1 &\geq \sum_{t \in [T]} \Pr[\cM(\bu') \in \cA^t] \\
&\overset{\eqref{eq:group-dp-application-prop}}{\geq} 0.5 \cdot e^{-\eps(2nc + 2)} \cdot T \\
&\geq 0.5 \cdot e^{-4nc\eps} \cdot \left\lfloor\frac{m}{4nc}\right\rfloor \\ 
&\geq 0.5 \cdot e^{-0.04\log (m/n)} \cdot \left\lfloor 5 \sqrt{m/n} \right\rfloor \\ 
&\geq (0.5 (m/n)^{-0.04}) \cdot (2.5 \sqrt{m/n}) > 1,
\end{align*}
where the third and fourth inequalities follow from our choice of parameters $c, T$ and the assumption that $c \geq 1$. 
This is a contradiction which establishes \Cref{thm:lb-agent-item-dp-connected-prop}.
\renewcommand\qedsymbol{$\square$}
\end{proof}

\section{Conclusion and Future Work}

In this paper, we have studied the fundamental task of fair division under differential privacy constraints, and provided algorithms and impossibility results for approximate envy-freeness and proportionality. 
There are several open questions left by our work. 
First, it would be useful to close the gaps in terms of $n$; for example, our envy-freeness upper bound for (agent $\times$ item)-level DP grows linearly in $n$ (\Cref{thm:alg-ef-em}) but our lower bound (\Cref{thm:lb-agent-item-dp-connected}) does not exhibit this behavior. Another perhaps more interesting technical direction is to extend our lower bounds for (agent $\times$ item)-level DP to arbitrary (i.e., not necessarily connected) allocations. 
Specifically, we leave the following intriguing open question: 
Is there an (agent $\times$ item)-level $\eps$-DP algorithm that, with probability at least $0.99$, outputs an EF$c$ allocation for $c = O_\eps(1)$ regardless of the values of $n$ and~$m$? 

While we have considered the original notion of DP proposed by \citet{DworkMcNi06}, there are a number of modifications that could be investigated in future work. 
A commonly studied notion is \emph{approximate DP} (also called \emph{$(\eps, \delta)$-DP}), which has an additional parameter $\delta \geq 0$ that specifies the probability with which the condition $\Pr[\cM(X) = o] \leq e^\eps \cdot \Pr[\cM(X') = o]$ is allowed to fail \citep{DworkKeMc06}. 
The notion of DP that we use in this paper corresponds to the case $\delta = 0$ and is often referred to as \emph{pure DP}. 
Several problems in the literature are known to admit approximate-DP algorithms with better guarantees compared to pure-DP algorithms (see, e.g., the work of \citet{SteinkeUl16}). 
In light of this, it would be interesting to explore whether a similar phenomenon occurs in fair division as well.

\section*{Acknowledgments}

This work was partially supported by the Singapore Ministry of Education under grant
number MOE-T2EP20221-0001 and by an NUS Start-up Grant.

\bibliography{aaai23}

\appendix

\section{Omitted Proofs}

\subsection{Proof of \Cref{thm:parallel-comp}}
\label{app:proof-parallel-comp}

Consider any pair of adjacent inputs $X \sim X'$, and any $o_1, \dots, o_k \in \range(\cM)$. 
From the condition on $\Gamma$, there exists $i^* \in [k]$ such that $\Gamma(X)_i = \Gamma(X')_i$ for all $i \ne i^*$ and $\Gamma(X)_{i^*} \sim \Gamma(X')_{i^*}$. 
From this, we have
\begin{align*}
\Pr[\cM'&(X) = (o_1, \dots, o_k)] \\
&= \Pr[\forall i \in [k], \cM(\Gamma(X)_i) = o_i] \\
&= \prod_{i \in [k]} \Pr[\cM(\Gamma(X)_i) = o_i] \\
&= \left(\prod_{i \in [k] \setminus \{i^*\}} \Pr[\cM(\Gamma(X)_i) = o_i]\right) \\
&\qquad \cdot \Pr[\cM(\Gamma(X)_{i^*}) = o_{i^*}] \\
&= \left(\prod_{i \in [k] \setminus \{i^*\}} \Pr[\cM(\Gamma(X')_i) = o_i]\right) \\
&\qquad\cdot \Pr[\cM(\Gamma(X)_{i^*}) = o_{i^*}] \\
&\leq \left(\prod_{i \in [k] \setminus \{i^*\}} \Pr[\cM(\Gamma(X')_i) = o_i]\right) \\
&\qquad \cdot e^\eps \cdot \Pr[\cM(\Gamma(X')_{i^*}) = o_{i^*}] \\
&= e^\eps\cdot \prod_{i \in [k]} \Pr[\cM(\Gamma(X')_i) = o_i] \\
&= e^\eps\cdot \Pr[\forall i \in [k], \cM(\Gamma(X')_i) = o_i] \\
&= e^\eps\cdot \Pr[\cM'(X') = (o_1, \dots, o_k)],
\end{align*}
where the inequality follows from the $\eps$-DP guarantee of $\cM$ and the fact that $\Gamma(X)_{i^*} \sim \Gamma(X')_{i^*}$.
Therefore, $\cM'$ is $\eps$-DP, as desired.

\subsection{Proof of \Cref{lem:littlewood-offord-simple}}
\label{app:proof-littlewood-offord}

Let $S := X_1 + \cdots + X_k$. Due to symmetry, we have
\begin{align*}
\Pr&\left[S < \frac{k}{2} - 0.1\sqrt{k}\right] \\
&= \frac{1}{2} \left(1 - \Pr\left[\left|S - \frac{k}{2}\right| \leq 0.1\sqrt{k}\right]\right).
\end{align*}
Let $b = \lfloor k/2 \rfloor$. We have
\begin{align*}
\Pr\left[\left|S - \frac{k}{2}\right| \leq 0.1\sqrt{k}\right] 
&= \sum_{i=\left\lceil k/2 - 0.1\sqrt{k}\right\rceil}^{\left\lfloor k/2 + 0.1\sqrt{k}\right\rfloor} \binom{k}{i}\cdot\frac{1}{2^k} \\
&\leq \sum_{i=\left\lceil k/2 - 0.1\sqrt{k}\right\rceil}^{\left\lfloor k/2 + 0.1\sqrt{k}\right\rfloor} \binom{k}{b}\cdot\frac{1}{2^k} \\
&\leq \left(0.2\sqrt{k} + 1\right)\binom{k}{b}\cdot\frac{1}{2^k} \\
&\leq 0.3\sqrt{k}\cdot\binom{k}{b}\cdot\frac{1}{2^k}.
\end{align*}
Now, we may bound $\binom{k}{b}\cdot\frac{1}{2^k}$ as follows:
\begin{align*}
\binom{k}{b}\cdot\frac{1}{2^k}
&= \frac{k!}{b!(k-b)!}\cdot\frac{1}{2^k} \\
&\overset{\eqref{eq:stirling}}{\leq} \left(\frac{e^{1/12}}{\sqrt{2\pi}} \cdot \frac{k^{k + 0.5}}{b^{b + 0.5}(k - b)^{(k - b) + 0.5}}\right) \cdot\frac{1}{2^k} \\
&= \frac{e^{1/12}}{\sqrt{\pi b}} \cdot \left(\frac{k}{2b}\right)^b \cdot \left(\frac{k}{2(k - b)}\right)^{(k - b) + 0.5} \\
&\leq \frac{e^{1/12}}{\sqrt{\pi b}} \cdot \left(\frac{k}{2b}\right)^b \\
&= \frac{e^{1/12}}{\sqrt{\pi b}} \left(1 + \frac{k-2b}{2b}\right)^b \\
&\leq \frac{e^{1/12}}{\sqrt{\pi b}} \left(1 + \frac{1}{2b}\right)^b \\
&\leq \frac{e^{1/12}}{\sqrt{\pi b}} \cdot e^{(1/2b) \cdot b} \\
&= \frac{e^{7/12}}{\sqrt{\pi b}} \\
&\leq \frac{1.6}{\sqrt{k}},
\end{align*}
where the third-to-last inequality follows from the well-known fact that $1+x \le e^x$ for every real number $x$, and the last inequality from $b \ge (k-1)/2$ and $k \geq 100$.

Combining the three (in)equalities above, we arrive at
\begin{align*}
\Pr\left[S < \frac{k}{2} - 0.1\sqrt{k}\right]
&\geq \frac{1}{2} \left(1 - 0.48\right) > \frac{1}{4},
\end{align*}
completing the proof.

\subsection{Proof of \Cref{lem:anti-concen}}
\label{app:proof-anti-concen}

Recall that a more precise version of Stirling's Formula (e.g.,~\cite{Robbins55}) gives
\begin{align} \label{eq:stirling}
\sqrt{2\pi}n^{n+0.5}e^{-n} \leq n! \leq e^{1/12} \cdot \sqrt{2\pi}n^{n+0.5}e^{-n}
\end{align}
for every positive integer~$n$.

Let $a = \lfloor k/2 + 0.1\sqrt{k \log \gamma} \rfloor + 1$, $b = a + \lfloor 0.2\sqrt{k} \rfloor$, and $r = b - k/2$. Note that by our choice of parameters, we have
\begin{align*}
(0.4\sqrt{\log \gamma} &- 0.2)\sqrt{k} \\
&\ge (0.4\sqrt{\log 2} - 0.2)\sqrt{k} 
\ge 0.1 \cdot 10 = 1.
\end{align*}
Hence, $0.2\sqrt{k} + 1 \le 0.4\sqrt{k\log\gamma}$,
and therefore
\begin{align}
r &\leq 0.1\sqrt{k \log \gamma} + 0.2\sqrt{k} + 1 \nonumber \\
&\leq 0.5\sqrt{k \log \gamma} \label{eq:bound-r} \\
&\leq 0.5 \sqrt{0.7k\log_2\gamma} \nonumber \\
&\leq 0.5\sqrt{0.7k(k/4)} < 0.21k, \nonumber
\end{align}
which means that $b < 0.71k$.

We have
\begin{align*}
&\Pr\left[S > \frac{k}{2} + 0.1\sqrt{k \log \gamma}\right] \\
&\geq \sum_{i=a}^b \Pr[S = i] \\
&= \sum_{i=a}^b \binom{k}{i} \cdot \frac{1}{2^k} \\
&\geq (b - a + 1) \cdot \binom{k}{b} \cdot \frac{1}{2^k} \\
&\geq 0.2\sqrt{k} \left(\binom{k}{b} \cdot \frac{1}{2^k}\right) \\
&= 0.2\sqrt{k} \left(\frac{k!}{b!(k-b)!} \cdot \frac{1}{2^k}\right) \\
&\overset{\eqref{eq:stirling}}{\geq} 0.2\sqrt{k} \left(0.3 \cdot \frac{k^{k + 0.5}}{b^{b + 0.5}(k - b)^{(k - b) + 0.5}} \cdot \frac{1}{2^k}\right) \\
&= 0.2\sqrt{k} \left(0.3\sqrt{2} \cdot \frac{1}{b^{0.5}} \cdot \left(\frac{k}{2b}\right)^b \cdot \left(\frac{k}{2(k - b)}\right)^{(k - b) + 0.5}\right) \\
&= 0.2\cdot 0.3\sqrt{2} \cdot \sqrt{\frac{k}{b}} \cdot \left(\frac{k}{2b}\right)^b \cdot \left(\frac{k}{2(k - b)}\right)^{(k - b) + 0.5} \\
&\geq 0.2\cdot 0.3\sqrt{2} \cdot \sqrt{\frac{1}{0.71}} \cdot \left(\frac{k}{2b}\right)^b \cdot \left(\frac{k}{2(k - b)}\right)^{(k - b) + 0.5} \\
&\geq 0.1\cdot\left(\frac{k}{2b}\right)^b \cdot \left(\frac{k}{2(k - b)}\right)^{k - b} \\
&= 0.1 \cdot \left(\frac{k}{2b}\right)^{2b-k} \cdot \left(\frac{k^2}{4b(k - b)}\right)^{k - b} \\
&\geq 0.1\cdot\left(\frac{k}{2b}\right)^{2b-k} \\
&= \frac{0.1}{\left(1 + \frac{2r}{k}\right)^{2r}} \\
&\geq \frac{0.1}{\exp(4r^2/k)} \\
&\overset{\eqref{eq:bound-r}}{\geq} \frac{0.1}{\gamma},
\end{align*}
where for the penultimate inequality we use the well-known fact that $1+x \le \exp(x)$ for every real number~$x$.
This completes the proof.

\end{document}